\documentclass[12pt]{article}
\usepackage[margin=1in]{geometry}
\usepackage{graphicx}
\usepackage{subcaption}
\graphicspath{{./plots/}}
\usepackage{float}
\usepackage{amsmath}
\usepackage[none]{hyphenat}
\usepackage[backref=page,hidelinks]{hyperref}
\renewcommand*{\backref}[1]{}
\renewcommand*{\backrefalt}[4]{%
	\ifcase #1 (Not cited.)%
	\or        (Cited on page~#2.)%
	\else      (Cited on pages~#2.)%
	\fi}
\usepackage[numbers]{natbib}
\usepackage{amssymb}
\usepackage{url} 

\usepackage{amsthm}
\newtheorem{prop}{Proposition}
\newtheorem{Def}{Definition}

\newtheorem*{Thm*}{Theorem}
\newtheorem{Egs}{Example}

\usepackage[nottoc,notlot,notlof]{tocbibind}
\usepackage[title,titletoc]{appendix}

\usepackage{makecell}

\usepackage{enumitem}

\usepackage[none]{hyphenat}

\newcommand{\R}{\mathbb{R}}
\newcommand{\E}{\mathbb{E}}
\renewcommand{\P}{\mathbb{P}}

\makeatletter
\g@addto@macro\bfseries{\boldmath}
\makeatother

\makeatletter
\newcommand*{\rom}[1]{\expandafter\@slowromancap\romannumeral #1@}
\makeatother

\usepackage[title]{appendix}

\begin{document}
	
		\title{Linear Regression for Power Law Distribution Fitting}
			\author{Samuel Forbes\footnote{forbessam2@gmail.com}}
		\date{December 2023}
		
		\maketitle
		
		\begin{abstract}
			We fit the exponent of the Pareto distribution, that is equivalent or can approximate  the continuous power law distribution given a cutoff point, using  linear regression (LR). We use LR on the logged variables of the empirical tail (one minus the empirical cumulative distribution function). We find the distribution of the consistent LR estimator and an approximate sigmoid relationship of the mean that underestimates the exponent. By factoring out a sigmoid function used to approximate the mean we transform the LR estimator so it is approximately unbiased with variance comparable to the minimum variance unbiased transformed MLE estimator. 
		\end{abstract}
		
		\section{Introduction}
		
		Power laws in the probability distribution appear across a wide number of disciplines \cite{clauset2009power}. We shall focus in this paper on \textbf{continuous power law distributions} that approximate variables such as income, wealth and stock returns \cite{gabaix2009power}. Several definitions of a power law distribution appear in the literature. We focus on the piecewise power law distribution and the regularly varying distribution which we define in the next section.  The power law part of the former distribution can be exactly viewed as a Pareto distribution whereas the power law part of the latter can only be approximated as a Pareto distribution. We fit only the exponent of the distribution and not the cutoff point which can be estimated using a Monte Carlo method found in Clauset et al. \cite{clauset2009power}. There are several methods of fitting the exponent of a Pareto distribution that are consistent (the estimator tends to the true parameter as the sample size tends to infinity) including the method of moments, maximum likelihood estimation (MLE), quantile methods and linear regression (LR)  \cite{quandt1964old}. Many have pointed out, for example \cite{clauset2009power},  that LR is inferior to MLE but to the author's knowledge it has not been analysed in detail why this is the case.
		We focus in this paper on fitting the exponent with LR on the empirical tail (one minus the empirical cumulative distribution funciton). We give a theoretical result on the distribution of the LR estimator as well as approximations on the mean and variance via simulation\footnote{Simulations are found at the author's github \cite{sam_git_code}}. We find evidence that the LR estimator is biased with a mean that underestimates the true exponent in a sigmoidal fashion depending on the sample size. Factoring out a sigmoid function we find this non-linear transformation of the LR estimator is roughly unbiased and comparable though greater in variance to a non-linear transformation of the MLE estimator which is unbiased of minimum variance \cite{rytgaard1990estimation}. We emphasise that the MLE estimator is superior however we present these results on the LR estimator as they may be of interest. 
		
		\section{Definitions of a Power Law Distribution}
		
		\label{sec_pl}
		
		As in Clauset et al. \cite{clauset2009power} we define the power law in two ways: one as a piecewise distribution and the other asymptotically as a regularly varying distribution. The piecewise power law distribution approximates the regularly varying distribution after a cutoff point which we assume is known. In this paper we will fit only the exponent of the piecewise power law distribution assuming the cutoff. We show below how this is identical to fitting the exponent of the Pareto distribution.
		
		\subsection{Piecewise Distribution}
		
		Let $X$ be a continuous random variable.
		We define the \textbf{power law distribution} as one that has a power law tail\footnote{We note Clauset et al. \cite{clauset2009power} defines the power law such that $\P(X>x) = \alpha/x^{\beta-1}$ so the density is $f(x) = c/x^{\beta}$ for an appropriate constant $c\,.$ }
		\begin{equation}
		\P(X>x) = \alpha/x^{\beta} \, , \quad  x>x_m \,
		\label{pl_tail}
		\end{equation}
		with $\alpha, \,\beta, \, x_m > 0\,.$ We call $\beta$ the \textbf{power law exponent} and $x_m$ the \textbf{cutoff}. Before $x_m$ we assume the tail of $X$ is defined by another function so that $X$ is a piecewise distribution.
		By taking the negative of the derivative of \eqref{pl_tail} we have that the density of $X$ is
		\begin{equation}
		f(x) = \beta \alpha/x^{\beta+1} \, , \quad  x>x_m \,.
		\label{dens_pl}
		\end{equation}
		Let us define the truncated distribution:
		\begin{equation}
		X_{\text{P}} = \{X|X>x_m\}\,,
		\label{Pareto_RV}
		\end{equation}
		then one can integrate over the entire domain, $[x_m,\infty)\,,$ to find that $X_{\text{P}}$ is a \textbf{Pareto distribution} with tail \eqref{pl_tail} and density \eqref{dens_pl} such that $\alpha = x_m^{\beta}\,.$ Now, for all $x>x_m\,,$
		\begin{align}
		\P(X>x) &= \P(X>x,X>x_m) = \P(X>x_m)\P(X>x|X>x_m) \nonumber \\ 
		&= \P(X>x_m)\P(X_{\text{P}}>x) \nonumber \\ 
		&=\P(X>x_m)\left(\frac{x_m}{x}\right)^{\beta} \,. \label{pl_pareto}
		\end{align}
		Thus in general $\alpha = \P(X>x_m) x_m^{\beta}\,.$
		  We shall assume \textbf{the cutoff} $x_m$ \textbf{is known (or has been fitted)} and we are concerned only  with \textbf{fitting the exponent} $\beta\,.$ We refer the reader to Clauset et al. \cite{clauset2009power} for a Monte-Carlo method using the Kolmogorov-Smirnov statistic to fit $x_m\,.$
		
		\subsection{Regularly Varying Distribution} 
		
		We shall briefly discuss a more general asymptotic representation of the power law distribution, the regularly varying distribution, and show how the piecewise distribution above approximates this distribution for large enough $x_m\,.$ We mention here that the method for fitting $x_m$ in Clauset et al. \cite{clauset2009power} has useful statistical properties such as consistency when applied to regularly varying distributions, see \cite{bhattacharya2020consistency}. 
		
		\begin{Def}[Regularly Varying Distribution, see e.g. \cite{jessen2006regularly}]
			A continuous random variable $X$ has a regularly varying distribution if
			for all $\lambda>0$ there exists a $\beta>0$ such that the tail distribution is a regularly varying function:
			\begin{equation*}
			\lim_{x \rightarrow \infty} \frac{\P(X> \lambda x)}{\P(X>x)} = \lambda^{-\beta}\,.
			\end{equation*}
		\end{Def}
		We see simply that the piecewise distribution \eqref{pl_tail} is a regularly varying distribution.
		
		\begin{Def}[Slowly Varying Function, see e.g. \cite{jessen2006regularly, seneta2006regularly}]
			A function $l$ is slowly varying if for all $\lambda>0$
			\begin{equation*}
			\lim_{x \rightarrow \infty} \frac{l(\lambda x)}{l(x)} = 1 \,.
			\end{equation*}
		\end{Def}
		
		One can show (p2 of \cite{seneta2006regularly}) that a regularly varying distribution can be written as 
		\begin{equation*}
		\P(X>x) = l(x) x^{-\beta}
		\end{equation*}
		where $l(x)$ is a slowly varying function. We have the following properties
		(p18 of \cite{seneta2006regularly}) of slowly varying functions: for all $\varepsilon>0$
		\begin{equation}
		\lim_{x \rightarrow \infty } l(x) x^{-\varepsilon} =0\,, \quad 
		\lim_{x \rightarrow \infty } l(x) x^{\varepsilon} = \infty\,.
		\label{slow_var_limits}
		\end{equation}
		
		We now show that a regularly varying distribution can be approximated by the piecewise power law distribution \eqref{pl_tail}.
		\begin{prop}
			If $X$ is a regularly varying distribution then for any $\varepsilon>0$ there exists an $x_m>0$ such that with $\alpha=\P(X>x_m) x_m^{\beta}$ and $x>x_m$ 
			\begin{equation*}
			\alpha x^{-\beta-\varepsilon}< \P(X>x)< \alpha x^{-\beta+\varepsilon} \,.
			\end{equation*}
			\label{reg_vary_approximation}
		\end{prop}
		\begin{proof}
			Suppose for any $\varepsilon>0\,,$ that there exists an $x_m>0$ such that for  $x>x_m$ and $\alpha =\P(X>x_m) x_m^{\beta}$ that
			\begin{align*}
			\alpha x^{-\beta-\varepsilon}&< \P(X>x)<\alpha x^{-\beta+\varepsilon} \Leftrightarrow \\
			\alpha x^{-\beta-\varepsilon}&< l(x)x^{-\beta}<\alpha x^{-\beta+\varepsilon} \Leftrightarrow \\
			\alpha x^{-\varepsilon}&< l(x)<\alpha x^{\varepsilon} 
			\end{align*}
			which is true by the definition of the above limits of slowly varying functions \eqref{slow_var_limits}.
		\end{proof}
		
		Thus Proposition \ref{reg_vary_approximation} shows that for $\varepsilon \gtrapprox 0 $ there exists an $x_m$ such that for $x>x_m$ we have $\P(X>x) \approx \alpha x^{-\beta}$ with $\alpha =\P(X>x_m) x_m^{\beta}\,.$
		
		\subsection{Examples}
		We now give an example of each of the two types of power law defined above. First a power law distribution that is piecewise (and regularly varying) that has an exponential then power law shape:
		\begin{Egs}[Piecewise power law, similar to (3.10) in \cite{clauset2009power}]
			For $x_m\,,\,\beta>0$ define the continuous random variable $X$ by the piecewise tail
			\begin{equation}
			\P(X>x) = \begin{cases}
			e^{-\beta \frac{x}{x_m}}\,, & 0<x \leq x_m \\
			e^{-\beta} \left(\dfrac{x_m}{x}\right)^{\beta}\,, & x>x_m
			\end{cases} \label{exp_pl_piecewise}
			\end{equation}
			\label{egs_exp_pl_piecewise}
		\end{Egs}
		Second a power law distribution that is not piecewise but is regularly varying:
		\begin{Egs}[Lomax \cite{lomax1954business} or equivalently Pareto Type \rom{2} \cite{arnold2008pareto}]
			For $ \lambda \,,\, \beta>0$ define the continuous random variable $X$ by the regularly varying tail
			\begin{equation}
			\P(X>x) = (1+x/\lambda)^{-\beta}\,, \quad x>0 \,.
			\end{equation}
			\label{egs_reg_var}
		\end{Egs}
			
		We now sample once from each of the distributions in the above examples and plot the empirical tail, see the appendix. We fit with the unbiased minimum variance MLE and roughly unbiased LR power law exponent estimators found in Sections \ref{sec_mle} and \ref{sec_lr}. For the piecewise distribution the power law part can be viewed exactly as points from a Pareto distribution. For the Lomax distribution a cutoff point $x_m$ must be chosen so that beyond $x_m$ we can only approximate by a Pareto distribution. We choose the cutoff exactly for the piecewise distribution and only approximately using the technique in Clauset et al. \cite{clauset2009power} for the Lomax distribution. We see the fits with both methods for the exponent of this particular Lomax distribution are not very accurate even for a large sample size. Thus we urge caution when using these methods for fitting the exponent of a regularly varying distribution (we find evidence from simulation not presented here that in this case a larger initial sample with a subsequently larger $x_m$ is required to get a closer estimate consistently).
	 	\begin{figure}[H]
	 		\centering
	 		\captionsetup{justification=centering}
	 		\includegraphics[width=0.8\textwidth]{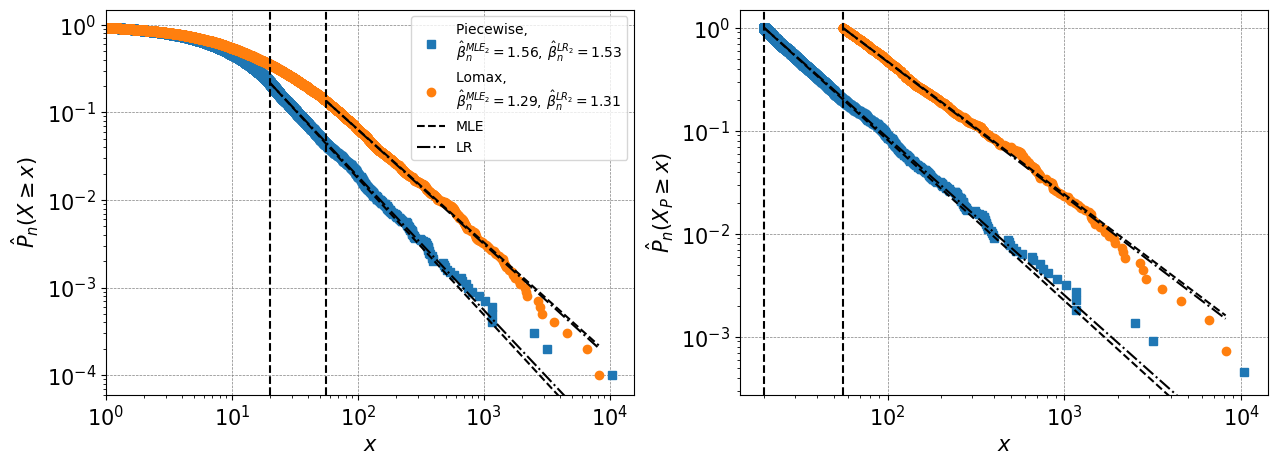}
	 		\caption{We generate a single sample of size $10^4$ from each distribution in Example \ref{egs_exp_pl_piecewise} and \ref{egs_reg_var}. For both examples we set $\beta=1.5$ and fit the power law points \eqref{pl_tail} to estimate $\beta\,$. For Example \ref{egs_exp_pl_piecewise} we set $x_m=20$ and fit the power law, of which there are roughly $2100$ points,  with both MLE \eqref{MLE_beta_est1}, fit $\hat{\beta}_n^{\text{MLE}}\approx1.56$, and LR \eqref{beta_lr_est2}, fit $\hat{\beta}_n^{\text{LR}}\approx1.53\,$. For Example \ref{egs_reg_var} we set $\lambda=20$ and estimate $x_m \approx 56$ using the method in \cite{clauset2009power} and fit the power law, of which there are roughly $1300$ points, with both MLE \eqref{MLE_beta_est1}, fit $\hat{\beta}_n^{\text{MLE}}\approx1.29$ and LR \eqref{beta_lr_est2}, fit $\hat{\beta}_n^{\text{LR}}\approx1.31\,$. Left figure is the full distribution, right figure is where we cut off the distribution at $x_m$ so that this new distribution $X_{\text{P}}$ \eqref{Pareto_RV} is exactly or approximated by a Pareto distribution.}
	 		\label{piecewise_Lomax}
	 	\end{figure}
		
		\section{Fitting the Power Law Exponent}
		
		We consider an i.i.d. sample $\mathbf{x}=\{x_1,x_2,\dots,x_n\}$ from a random variable $X$ that is from the power law tail part of the distribution \eqref{pl_tail}. We assume $x_m$ and $\P(X>x_m)$ are known or have been fitted and emphasise $x_m$ is not from the sample but is the cutoff in \eqref{pl_tail} for which $x_i>x_m$ for all $x_i\in \mathbf{x}\,$. We note from \eqref{pl_pareto} we only have to fit the Pareto  distribution and thus take $\mathbf{x}$ from the Pareto distribution $X_{\text{P}}$ \eqref{Pareto_RV}. 
		
		\subsection{Maximum Likelihood Estimation}
		
		\label{sec_mle}
		
			Suppose $X$ follows a power law distribution defined by \eqref{pl_tail}. We saw that 
			$X_{\text{P}}$ follows a Pareto distribution with density 
			$f(x|\beta)= \beta x_m^{\beta}x^{-(\beta+1)}\,.$ Assume $\mathbf{x}=\{x_1,x_2,\dots,x_n\}$ is an i.i.d. sample from $X_{\text{P}}$ then the log-likelihood is defined
			\begin{equation*}
			\log L(\mathbf{x}|\beta) = \log \left(\prod_{i=1}^{n}f(x_i|\beta) \right) \,.
			\end{equation*}	
			By maximising the log-likelihood (respectively the likelihood) with respect to $\beta$ we find that the MLE estimator for $\beta$ given $x_m$ is
			\begin{equation}
			\hat{\beta}_n^{\text{MLE}_1} = \frac{1}{1/n \sum_{i=1}^{n}\log (x_i/x_m) } \,.
			\label{MLE_beta_est}
			\end{equation}
			
			The MLE estimator \eqref{MLE_beta_est} for the Pareto distribution has been known for a long time with Clauset et al. \cite{clauset2009power} referencing Muniruzzaman, 1957 \cite{muniruzzaman1957measures} as an early paper containing this result. We refer the reader to \cite{clauset2009power} for a summary of useful statistical properties of this estimator which includes consistency. 
		 We note that \eqref{MLE_beta_est} is exactly the Hill's estimator for the Pareto distribution $X_{\text{P}}$ \cite{hill1975simple}. In particular it can be shown, see Rytgaard \cite{rytgaard1990estimation}, that 
		 	$\hat{\beta}_n^{\text{MLE}_1}$ has an inverse Gamma distribution with mean and variance
		 	\begin{equation}
		 	\E[\hat{\beta}_n^{\text{MLE}_1}]=n \beta/(n-1)\,, \quad \text{Var}(\hat{\beta}_n^{\text{MLE}_1})= n^2 \beta^2 /((n-1)^2(n-2))\,.
		 	\label{mle_mean_var}
		 	\end{equation}
			
			This MLE  estimator \eqref{MLE_beta_est}, as noted in \cite{rytgaard1990estimation}, can be corrected for bias. From \eqref{mle_mean_var} we see the \textbf{unbiased MLE estimator} is (for $n>1$)
			\begin{equation}
			\tilde{\beta}_n^{\text{MLE}_2}= \frac{n-1}{n}\hat{\beta}_n^{\text{MLE}_2} = 
			\frac{1}{\frac{1}{n-1} \sum_{i=1}^{n}\log (x_i/x_m) }
			\label{MLE_beta_est1}
			\end{equation}
			with mean and variance (for $n>2$)
			\begin{equation}
			\E[	\hat{\beta}_n^{\text{MLE}_2}] = \beta \,, \quad 
			\text{Var}(\hat{\beta}_n^{\text{MLE}_2}) = \frac{\beta^2}{n-2}< \text{Var}(\hat{\beta}_n^{\text{MLE}_1}) \,.
			\label{mle1_mean_var}
			\end{equation}
			It is shown in \cite{rytgaard1990estimation} that the unbiased MLE estimator \eqref{MLE_beta_est1} has \textbf{minimum variance} and is \textbf{asymptotically normally distributed} with mean $\beta$ and variance $\beta^2/n$ and is a \textbf{consistent estimator} (the same for the untransformed MLE estimator \eqref{MLE_beta_est} \cite{clauset2009power}). 
			
			Now approximating either MLE estimator with the asymptotic normal we have that $\approx 68$ \% of values of the estimator lie within one standard deviation i.e. lie within the interval 
			\begin{equation}
			(\beta-\beta/\sqrt{n}, \beta + \beta/\sqrt{n}) \,.
			\label{MLE_interval}
			\end{equation}
	 Thus for the size of this interval to be less than $\varepsilon>0$ we want 
			\begin{equation*}
			\frac{2 \beta}{\sqrt{n}} < \varepsilon  \Leftrightarrow n> 4 \beta^2/ \varepsilon^2 \,.
			\end{equation*}
			Therefore for a majority of the values of the MLE estimator to be within for example $0.1$ of either side of $\beta$ we would want a sample of size at least $400 \beta^2\,.$ Now for many real-world situations $0.5<\beta<3\,,$ see e.g. Table 6.1 in \cite{clauset2009power}, thus the sample is required to be quite large before one is reasonably confident that the minimum variance MLE estimator is close to $\beta\,.$

		\subsection{Linear Regression}
		\label{sec_lr}
			We shall now apply linear regression to fit the power law exponent to the power law model \eqref{pl_tail} using the empirical tail (one minus the empirical distribution function, see the appendix). Let us again assume $\mathbf{x}=\{x_1,x_2,\dots,x_n\}$ is an i.i.d. sample from $X_{\text{P}}$ defined above. Then for any $x_i \in \mathbf{x}\,$ the empirical tail is
			\begin{equation}
			\hat\P_n(X_{\text{P}} \geq x_i) = (x_m/x_i)^{\beta} +\varepsilon_{n,x_i} \,
			\label{pl_error1}
			\end{equation}
			where the error $\varepsilon_{n,x_i}$ has the following properties (see the appendix):
			\begin{equation}
			\varepsilon_{n,x_i} \sim (1/n)\text{Binom}_{k\geq1}(n,(x_m/x_i)^{\beta})-(x_m/x_i)^{\beta}\,,
			\label{pl_error2}
			\end{equation}
			where $\text{Binom}_{k\geq1}(.)$ is the truncated Binomial distribution that restricts the support to $\{1,2,\dots,n\}$, for large $n$
			\begin{equation*}
			\varepsilon_{n,x_i} \approx \mathcal{N}(0,(1/n)(x_m/x_i)^{\beta}(1-(x_m/x_i)^{\beta}))\,
			\end{equation*}
			and 
			\begin{equation}
			\varepsilon_{n,x_i} \rightarrow 0 \quad \text{almost surely as } n \rightarrow \infty\,. \label{error_lin_reg_limit}
			\end{equation}
			
			Applying log's to \eqref{pl_error1} we have
			\begin{align}
			\log \hat\P_n(X_{\text{P}} \geq x_i) &= \beta \log x_m - \beta \log x_i + \log (1+(x_m/x_i)^{-\beta} \varepsilon_{n,x_i}) \nonumber \\
			&= - \beta \log (x_i/x_m) +\tilde{\varepsilon}_{n,x_i} \label{linear_eq}
			\end{align}
			Where the new error is
			\begin{equation}
			\tilde{\varepsilon}_{n,x_i}=\log (1+(x_m/x_i)^{-\beta} \varepsilon_{n,x_i})\,.
			\label{log_reg_error}
			\end{equation}
			We note from \eqref{pl_error2} that
			\begin{equation}
			\tilde{\varepsilon}_{n,x_i} \sim \log B_{n,x_i}
			\label{log_reg_error1}
			\end{equation}
			where $B_{n,x_i} \sim (1/n)(x_m/x_i)^{-\beta}\text{Binom}_{k\geq1}(n,(x_m/x_i)^{\beta})\,.$

			Applying \textbf{ordinary least squares}, see e.g. Section 3.1 of  \cite{seber2003linear}, to \eqref{linear_eq} the \textbf{LR estimator} for $\beta$ can be found to be
			\begin{equation}
			\hat{\beta}^{\text{OLS}_1}_n = \frac{\sum_{i=1}^{n}\log(x_i/x_m)	\log(1/ \hat\P_n(X_{\text{P}} \geq x_i))}{\sum_{i=1}^{n}(\log(x_i/x_m))^2} \,.
			\label{beta_lr_est}
			\end{equation}
			We note that if $a\,,b \in \R_{>0}$ then by the change of base formula for logs that
			\begin{equation*}
			\log_{a}x=\log_{b}x/\log_{b}a
			\end{equation*}
			 and so due to cancellation it does not matter which base we take for the log in \eqref{beta_lr_est}. As the $\tilde{\varepsilon}_{n,x_i}$ are a function of the independent $x_i$ then they are independent. However as we will see from simulation, the mean of \eqref{beta_lr_est} is biased, indicating that the errors $\tilde{\varepsilon}_{n,x_i}$ do not satisfy the properties of zero mean and constant variance to satisfy the further Gauss-Markov conditions, see e.g. Theorem 3.2 of \cite{seber2003linear}. Due to the relatively complicated distribution of $\tilde{\varepsilon}_{n,x_i}$ \eqref{log_reg_error1} it is however hard to prove these properties formally.
			However as $n \rightarrow \infty$ we have by \eqref{error_lin_reg_limit} and
			\eqref{log_reg_error} that
			\begin{equation*}
			\tilde{\varepsilon}_{n,x_i} \rightarrow 0 \quad \text{almost surely as } n \rightarrow \infty\,, \label{error_lin_reg_limit1}
			\end{equation*}
			thus the estimator will tend almost surely to $\beta$ as $n \rightarrow \infty$ and so the  
			LR estimator \eqref{beta_lr_est} is a \textbf{consistent estimator} as noted by 
			Quandt \cite{quandt1964old}.

			Now let $\mathbf{x}_{\text{o}}=\{x_{(1)},x_{(2)},\dots,x_{(n)}\}=\mathbf{x}$ be the ordered sample of $\mathbf{x}$ thus \newline $x_{(1)}<x_{(2)}<\dots<x_{(n)}\,.$ We have that the empirical tail, see the appendix, of the ordered sample is
			\begin{equation}
			\hat{\P}_n(X \geq x_{(i)}) = \frac{n-i+1}{n}\,,\quad i=1,2,\dots,n\,.
			\label{emp_tail_order}
			\end{equation}
			As $x_i$ is from a Pareto distribution with parameters $x_m$ and $\beta$ then $k_i=\log(x_i/x_m)$ (specifically for the \textbf{natural logarithm} which we use from now on) is exponentially distributed with parameter $\beta\,$. Now the result by R\'enyi, see \cite{renyi1953theory}, on the distribution of the order statistics for an exponential distribution gives 
			\begin{equation}
			k_{(i)}=\log(x_{(i)}/x_m) \sim \frac{1}{\beta} \sum_{j=1}^{i} \frac{Z_j}{n-j+1}
			\,,\quad i=1,2,\dots,n \,
			\label{order_exp}
			\end{equation}
			where $Z_j \sim \text{Exp}(1)\,$ (exponentially distributed with parameter $1\,$).
			Substitution of \eqref{emp_tail_order} and \eqref{order_exp} into the the OLS estimator \eqref{beta_lr_est} leads to
			\begin{equation}
			\hat{\beta}_n^{\text{OLS}_1} \sim  \left(\frac{\sum_{i=1}^{n}\sum_{j=1}^{i}\frac{Z_j}{n-j+1} \log\left(\frac{n}{n-i+1}\right)}{\sum_{i=1}^{n}
				\left(\sum_{j=1}^{i}\frac{Z_j}{n-j+1}\right)^2} \right) \beta \,.
			\label{beta_lr_est1}
			\end{equation}
			Thus we have a representation of the distribution of the LR estimator \eqref{beta_lr_est}. It is a rather complicated formula which is beyond the scope of this paper to analyse fully so we make do with analysing the LR estimator via simulation. We note though that the term in the brackets of \eqref{beta_lr_est1} does not depend on $\beta\,,$ thus the distribution of the LR estimator is a random factor multiplied by $\beta\,.$
			From simulations, see Figure \ref{beta_fits_lr}, we hypothesise two results:
			\begin{align}
			\E[\hat{\beta}_n^{\text{OLS}_1}] &\approx  \log\left(e - (\log n)^{\gamma}/n\right) \beta\,, \quad \gamma \approx 1.6 \label{mean_ols} \\
			\text{Var}(\hat{\beta}_n^{\text{OLS}_1}) &\approx \mathcal{O}(1/n) \label{var_ols}
			\end{align}
			where $e=2.718\dots$ is Euler's number. The free parameter $\gamma$ in \eqref{mean_ols} was approximated using non-linear least squares. Let 
			\begin{equation*}
			r_n=\log\left(e - (\log n)^{\gamma}/n\right)\,, \quad \gamma \approx 1.6 \,.
			\end{equation*}
			This function we propose is an approximation to the mean of the random factor in the brackets in \eqref{beta_lr_est1}.
			We see $ r_n < 1$ for finite $n$ and $r_n \rightarrow 1 $ as $n \rightarrow \infty$ so that 
			 \eqref{mean_ols} underestimates $\beta$ but tends to $\beta$ as $n \rightarrow \infty\,.$ We note that the mean approximation \eqref{mean_ols} is \textbf{sigmoidal} but imperfect especially for small $n$ and further simulation would be needed to check if it is still a good approximation above $n=1000\,,$ see Figure \ref{beta_fits_lr}. It would be future work to determine more about the form of \eqref{beta_lr_est1} and whether it has an analytic mean or variance. If an analytic mean exists then a strictly unbiased LR estimator can be found.
			 	\begin{figure}[H]
			 		\centering
			 		\captionsetup{justification=centering}
			 		\includegraphics[width=0.8\textwidth]{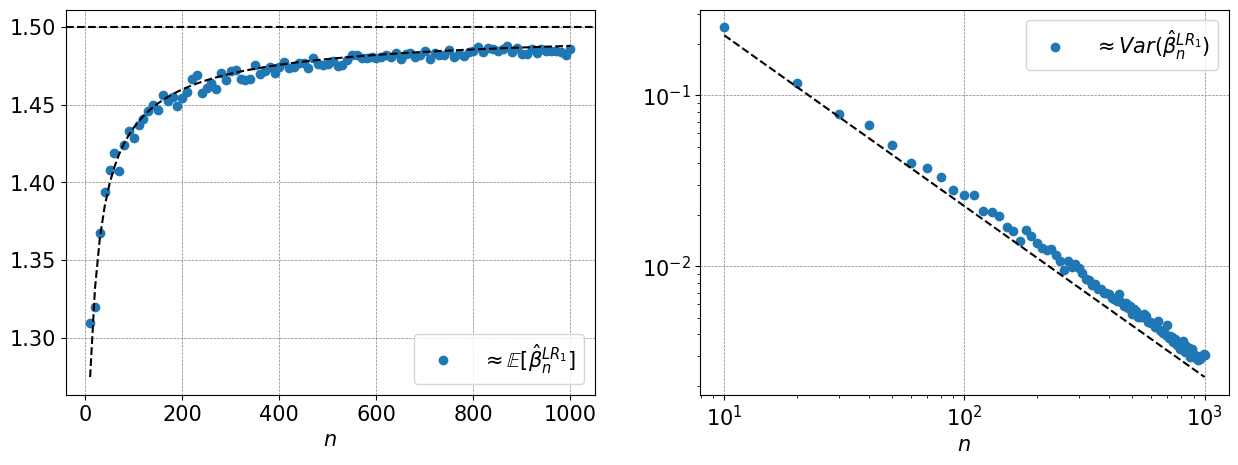}
			 		\caption{Sample mean, left and sample variance, right of the biased LR estimator $\hat{\beta}_n^{\text{OLS}_1}$ \eqref{beta_lr_est} for a Pareto distribution $X_{\text{P}}$ with $\beta=1.5$ and known $x_m=1$ of 1000 samples of size $n=10,20,\dots,1000\,.$ Black dashed line in left is the mean approximation \eqref{mean_ols} and black dashed line in right is $\beta^2/n$ the asymptotic minimum variance of the MLE estimator.}
			 		\label{beta_fits_lr}
			 	\end{figure}
			 
			 We can though still approximately correct for the bias to obtain the approximately unbiased, consistent LR estimator
			\begin{equation}
			\hat{\beta}_n^{\text{OLS}_2}=\frac{\hat{\beta}_n^{\text{OLS}_1}}{r_n}
			\label{beta_lr_est2}
			\end{equation}
			with
			\begin{equation}
			\E[\hat{\beta}_n^{\text{OLS}_2}] \approx \beta\,,
			\quad
			\text{Var}(\hat{\beta}_n^{\text{OLS}_2}) \approx \mathcal{O}(1/n)\,.
			\label{beta_lr_est2_mean_var}
			\end{equation}
			We note though that this transformation is greater in terms of variance:
			\begin{equation*}
			\text{Var}(\hat{\beta}_n^{\text{OLS}_2})= \frac{\text{Var}(\hat{\beta}_n^{\text{OLS}_1})}{r_n^2}>\text{Var}(\hat{\beta}_n^{\text{OLS}_1}) \,.
			\end{equation*}
			However we find through simulation, see Figure \ref{beta_fits_lr} and Figure \ref{beta_fits} in the next section, that the variance of both LR estimators are not too much larger than the minimum variance of the unbiased MLE estimator.
			
		\subsection{Comparison Between LR and MLE with Simulations}
		\label{sec_sim}

		In Figure \ref{beta_fits} we compare the approximately unbiased LR estimator \eqref{beta_lr_est2} with the minimum variance unbiased MLE estimator \eqref{MLE_beta_est1} on 1000 samples with sample sizes of $10,20,\dots,1000.$ We show in Figure \ref{beta_fits} that although the transformed LR estimator is worse than the transformed MLE estimator it is comparable.
		
		\begin{figure}[H]
			\centering
			\captionsetup{justification=centering}
			\includegraphics[width=0.8\textwidth]{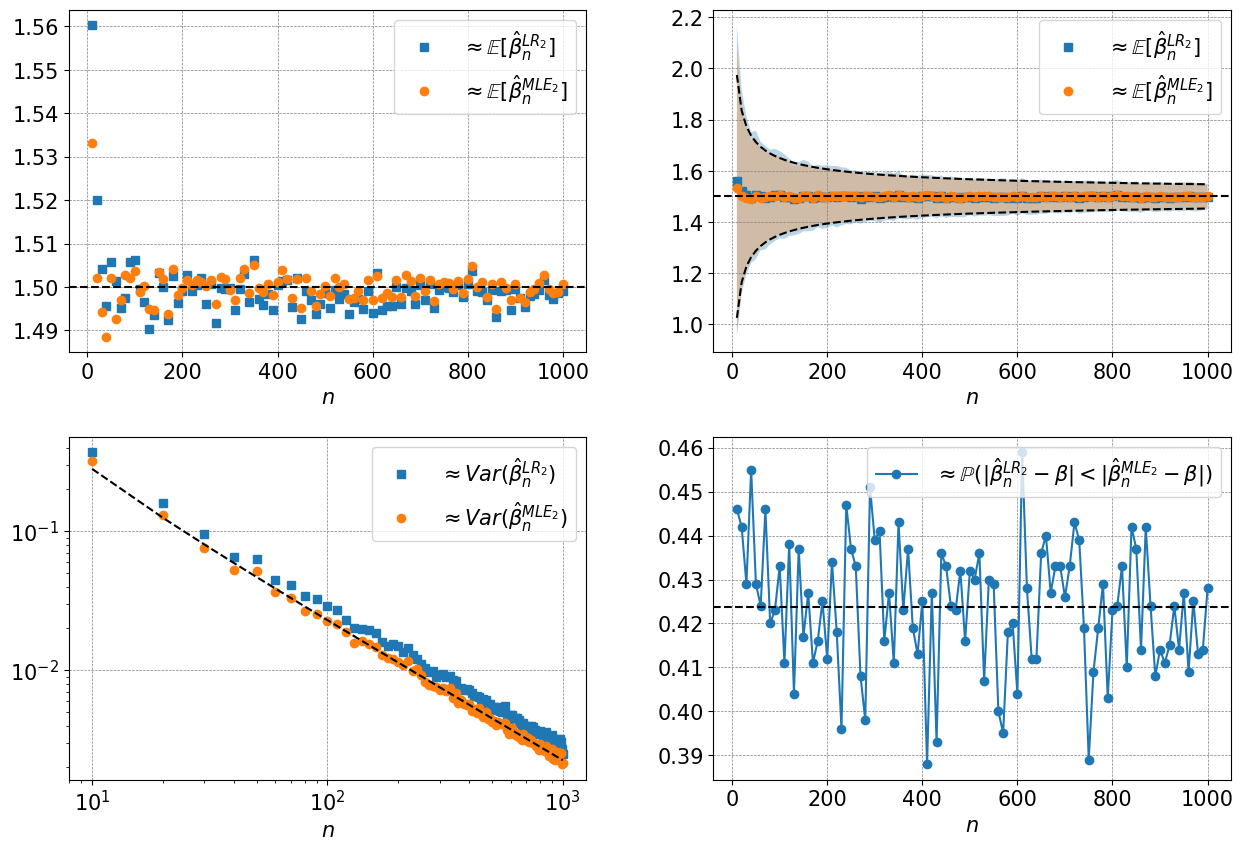}
			\caption{The approximately unbiased LR estimator $\hat{\beta}^{\text{LR}_2}_n$ \eqref{beta_lr_est2} and the unbiased minimum variance MLE estimator $\hat{\beta}^{\text{MLE}_2}_n$ \eqref{MLE_beta_est1} for $\beta$ were found for a $1000$ samples of size $n=10,20,\dots,1000$ from a Pareto distribution $X_{\text{P}}$ with known $x_m=1$ and unknown power law exponent $\beta\,.$ In this case we are estimating $\beta=1.5\,.$ Top left shows the estimator sample means against $\beta$, top right shows the sample means within one sample standard deviation and bounded by the interval \eqref{MLE_interval}, bottom left shows the sample variances where the dotted black line is the transformed MLE estimator variance \eqref{mle1_mean_var} of $\mathcal{O}(1/n)$ and bottom right shows approximately the probability the LR estimator is closer to $\beta$ than the transformed MLE estimator which fluctuates roughly around $0.42\,$.  }
			\label{beta_fits}
		\end{figure}

		\section{Conclusion}
		
		We use linear regression on the empirical tail to fit the exponent of the Pareto distribution that approximates or is identical to a continuous power law distribution given a cutoff value. We give an analytical representation of the distribution of the LR estimator and show evidence it's mean underestimates the true exponent in a sigmoidal fashion. By factoring out a sigmoid function of the LR estimator we find comparable results to the transformed MLE estimator. Although the minimum variance unbiased transformed MLE estimator is superior we nonetheless present these novel results for the LR estimator.
		
		\bibliographystyle{apalike} 
		\bibliography{mybib}
		
		\begin{appendices}
				
			\section*{Appendix: the Empirical Tail}
			
			Suppose $\mathbf{x}=\{x_1,x_2,\dots,x_n\}$ is an i.i.d. random sample from any  random variable $X$. 
			The \textbf{empirical cumulative distribution function} is defined, see e.g. Section 19.1 of \cite{van2000asymptotic}
			\begin{equation*}
			\hat{\P}_n(X < x) = \frac{1}{n}\sum_{i=1}^{n}\mathbf{1}_{x_i < x}
			\end{equation*}
			where $\mathbf{1}_A = \begin{cases}
			1 \quad \text{ if } x \in A \\
			0 \quad \text{ if } x \notin A
			\end{cases}$ is the indicator function. We define the \textbf{empirical tail} as
			\begin{equation}
			\hat{\P}_n(X \geq x) = 1- \hat{\P}_n(X < x) = \frac{1}{n}\sum_{i=1}^{n}\mathbf{1}_{x_i \geq x} \,.
			\label{empirical tail}
			\end{equation}
			which takes values $0,1/n,2/n,\dots,1\,.$
			We also introduce the error of the empirical tail $\varepsilon_n$
			\begin{equation*}
			\hat{\P}_n(X \geq x) = \P(X>x) + \varepsilon_n \,.
			\end{equation*}
			We now state some results which follow from known results on the empirical cumulative distribution function, see e.g. Section 19.1 of \cite{van2000asymptotic}. We have 
			\begin{equation*}
			\hat{\P}_n(X \geq x) \sim (1/n)\text{Binom}(n, \P(X \geq x) ) \,.
			\end{equation*}
			with mean and variance
			\begin{equation*}
			\E[\hat{\P}_n(X \geq x)] = \P(X \geq x)\,, \quad \text{Var}(\hat{\P}_n(X \geq x))=(1/n)\P(X \geq x)\P(X < x)\,.
			\end{equation*}
			By the law of large numbers
			\begin{equation*}
			\hat{\P}_n(X \geq x) \rightarrow \P(X \geq x) \quad \text{almost surely as}\quad n \rightarrow \infty \,.
			\end{equation*}
			By the central limit theorem 
			\begin{equation*}
			\sqrt{n}(\hat{\P}_n(X \geq x) - \P(X \geq x))
			\rightarrow \mathcal{N}(0,\P(X \geq x)\P(X < x)) \,
			\end{equation*}
			in distribution. 
			
			From the above we have the following properties for $\varepsilon_n \,:$
			\begin{equation*}
			\varepsilon_n \sim (1/n)\text{Binom}(n, \P(X \geq x) ) - \P(X \geq x)
			\end{equation*}
			with mean and variance
			\begin{equation*}
			\E[\varepsilon_n] = 0\,, \quad \text{Var}(\varepsilon_n)=(1/n)\P(X \geq x)\P(X < x) \,,
			\end{equation*}
			\begin{equation*}
			\varepsilon_n \rightarrow 0 \quad \text{almost surely as}\quad n \rightarrow \infty
			\end{equation*}
			and for large $n$ we can approximate the error
			\begin{equation*}
			\varepsilon_n \approx \mathcal{N}(0,(1/n)\P(X \geq x)\P(X < x))\,.
			\end{equation*}
			
		   Let $m = \min(\mathbf{x})$ and $M=\max(\mathbf{x})$ be respectively the min and max of the sample $\mathbf{x}\,.$ Also let $\mathbf{x}_{\text{o}}=\{x_{(1)},x_{(2)},\dots,x_{(n)}\}=\mathbf{x}$ be the ordered sample of $\mathbf{x}$ so that \newline $m=x_{(1)}<x_{(2)}<\dots<x_{(n)}=M\,.$ Then we have that 
		   	\begin{equation*}
		   	\hat{\P}_n(X \geq x_{(i)}) = \frac{n-i+1}{n} \in [1/n,1]\,,\quad i=1,2,\dots,n\,.
		   	\end{equation*}
		   and in particular for
		    $x \in [m,M]$ 
			\begin{equation*}
			\hat{\P}_n(X\geq x) \in [1/n,1]\,,\quad i=1,2,\dots,n\,.
			\end{equation*}
			Thus for this case the empirical tail with $x \in [m,M]$ will not be zero and so 
			\begin{equation*}
			\hat{\P}_n(X \geq x) \sim (1/n)\text{Binom}_{k \geq 1}(n, \P(X \geq x) )
			\end{equation*}
			where $\text{Binom}_{k \geq 1}(n, \P(X \geq x) )$ is the truncated Binomial distribution, see e.g. Section 3.11 of \cite{johnson2005univariate}, excluding values that take $0$ and thus with support $\{1,2,\dots,n\}\,.$

		\end{appendices}

\end{document}